\documentclass[12pt,journal]{IEEEtran}
\usepackage{amsmath}
\usepackage{amsfonts}
\usepackage{graphicx}

\newcommand{\e}{{\beta_2}}

\newtheorem{mytheorem}{Theorem}
\newtheorem{mylemma}{Lemma}

\newtheorem{myprop}{Proposition}

\newtheorem{myfact}{Fact}

\begin{document}
\title{Noise threshold for universality of 2-input gates}
\author{Falk~Unger
\thanks{F. Unger works at the Centrum voor Wiskunde en Informatica in Amsterdam, The Netherlands. (Falk.Unger@cwi.nl)
 }}
\maketitle

\begin{abstract} 
It is known that $\epsilon$-noisy gates with 2 inputs are universal for arbitrary computation (i.e.\ can compute any function with bounded error), if all gates fail independently  with probability $\epsilon$ and $\epsilon<\e=(3-\sqrt{7})/4\approx 8.856\%$.
In this paper it is shown that this bound is tight for formulas, by proving that gates with 2 inputs, in which each gate fails with probability at least $\e$ cannot be universal. Hence, there is a threshold on the tolerable noise for formulas with 2-input gates and it is $\e$. It is conjectured that the same threshold also holds for circuits. 
 \end{abstract}




\section{Introduction}
\label{sec:intro}
During the last decades, computers have become faster and faster, mainly due to advances in hardware miniaturization. However, there are physical limits to the possible extent of this miniaturization, and the closer one gets to these limits, the less robust and more error-prone the components become \cite{i:noisy,s:noisy}. It is estimated that the time when processor architects face these limitations is within the next decade \cite{Bobgodwell}. 

Gates, the smallest components of any processor, can fail in (at least) two ways. The first is that they do not work at all. The second is that they work most of the time correctly, and fail sometimes. This type of errors is called ``soft errors'' by hardware engineers. We deal with faults of the second type. 

In particular, we consider the computational model of \emph{noisy formulas}. Formulas are a special kind of circuits in which each gate has exactly one output wire \footnote{Precise definitions for all terms used can be found in Section \ref{Sec:prelim}.}. We ask how much noise on the gates is tolerable, such that any function can still be computed by some formula with bounded-error. We will assume throughout  that gates fail independently of each other.

This question has been studied earlier. Already in 1956 von Neumann discovered that reliable computation is possible with noisy 3-majority gates if each gate fails independently with probability less than $0.0073$ \cite{vN:reliable_organisms}.  The first to prove an upper bound on the tolerable noise was Pippenger \cite{p:noisy3}. He proved that formulas with gates of fan-in at most $k$, where each gate fails independently with probability at least $\epsilon \ge \frac{1}{2}-\frac{1}{2k}$, are not sufficient for universal computation (i.e.\ not all functions can be computed with bounded error). Feder proved that this bound also applies to circuits \cite{f:noisy}. Later, Feder's bound was improved to $\frac{1}{2}-\frac{1}{2\sqrt{k}}$ by Evans and Schulman \cite{es:noisy4}. 

For formulas with gates of fan-in $k$ and $k$ odd, Evans and Schulman \cite{es:noisyk} proved the tight bound $\beta_k=\frac{1}{2}-\frac{2^{k-2}}{k {k-1 \choose k/2-1/2}}$. Tight  here means that if all gates fail independently with the same fixed probability $\epsilon<\beta_k$, then any function can be bounded-error computed, and if each gate fails with some probability at least $\beta_k$ (which does not need to be the same for all gates), universal computation is not possible. For $k=3$  the threshold was first established by Hajek and Weller \cite{hw:noisy5}.

However, so far it has not been possible to establish thresholds for gates with \emph{even} fan-in (or even prove their existence), as pointed out in \cite{es:noisyk}. In particular, the most basic case of fan-in 2, which is most commonly used, had been elusive. An intuitive argument why even fan-in is different is that for even fan-in threshold gates (and in particular majority gates) can never be ``balanced'', in the sense that the number of inputs on which they are  $1$ cannot be the same as the number of inputs on which they are $0$.

Evans and Pippenger \cite{ep:noisy2} made some progress in this direction. First, they show that all functions can be computed by formulas with noisy NAND-gates with fan-in 2, if each NAND-gate fails with probability exactly $\epsilon$, for any $0\le \epsilon<\e=\frac{3-\sqrt{7}}{4}$. Second, they show that with NAND-gates alone this bound cannot be improved (They make some additional assumptions which we discuss below).
This left open the question of what the bound is if we allow all 16 gates with fan-in 2. We settle this question in this paper.
\begin{mytheorem}
\label{the:main}
Assume $\Delta >0$. Functions that are computable with bias $\Delta$ by a formula in which all gates have fan-in at most 2 and fail independently with probability at least $\e =(3 - \sqrt{7})/4$, depend on at most a constant number of input bits.
\end{mytheorem}

Together with the first mentioned result from \cite{ep:noisy2} this gives the exact threshold for formulas with gates of fan-in 2. It extends the second result from  \cite{ep:noisy2} in the following ways: (1) We allow all gates of fan-in 2, instead of only NAND-gates. (2) We prove that if the noise is exactly $\e$, then no universal bounded-error computation is possible. (3) In contrast to our result, the upper bound in \cite{ep:noisy2} only applies to ``soft'' inputs. They show that gates with noise more than $\e$ cannot increase the bias. More precisely, if the inputs to the formula are noisy themselves and have bias at most $\Delta>0$, then the output of the formula cannot have larger bias than $\Delta$. This left open the case where the input bits are not noisy and either $0$ or $1$, which is the case we care about most. Our argument shows that even with perfect inputs fault-tolerant computation is not possible for noise at least $\e$.

To prove Theorem \ref{the:main} we introduce a new technique, which is also applicable in the case of fan-in $2$. We expect that it can be extended to other fan-in cases.

We conjecture that our bound also holds for circuits.

\subsection{Outline of the proof}
\label{sec:outline}
For any function $f : \{0,1\}^n \rightarrow \{0,1\}$ we will choose an input bit $x_i$ which $f$ depends on, and fix all other bits such that $f$ still depends on $x_i$ . Assume that $f$ is computed by a formula $F$ with noisy gates that fail independently with probability at least $\beta_2$. 
Then, for each gate in the formula $F$ with input wires $A$ and $B$ and output wire $C$ we can define $a=\frac{1}{2}\mathbb{P}[A=0 \mid x_i=0]+\frac{1}{2}\mathbb{P}[A=0 \mid x_i=1]$ and $\delta_a=\mathbb{P}[A=0 \mid x_i=0]-\mathbb{P}[A=0 \mid x_i=1]$ and analogously for $B$ and $C$. The variable $a$ can be seen as the average probability of $A$ being $0$. We call $\delta_a$ the \emph{bias} of $A$ with respect to the two input settings $x_i=0$ and $x_i=1$. 

To prove our result one could attempt the following, which will turn out to not quite work (but we then show how to fix that): For an $\epsilon$-noisy gate with fan-in $2$, input wires $A$, $B$ and output wire $C$, we would like to show that if the noise $\epsilon$ is at least the threshold $\e$ then for any $\delta>0$ there is some $0\le \theta <1$ such that if $\delta\le\max\{ |\delta_a|,|\delta_b|\}$ then
\begin{equation}
\label{wrong}
|\delta_c| \le \theta \max\{|\delta_{a}|,|\delta_{b}|\}
\end{equation}
This would mean that the bias goes down exponentially with the number of computation steps, until it reaches $\delta$. Further, it is easy to show that for any $d>0$ there is a function $f$ such that any formula computing $f$ has one input bit $x_i$ on which $f$ depends and the number of computation steps on any path from $x_i$ to the output bit is at least $d$. Hence, the bias cannot be bounded away from zero for all $f$ and $x_i$. 


Unfortunately, (\ref{wrong}) is not always true. Sometimes the bias can actually go up.\footnote{An easy example is an OR-gate with noise $\epsilon=1/10$, $\delta_a=\delta_b=1/10$ and $a=b=8/10$, for which $\delta_c=(a\delta_b+b\delta_a)(1-2\epsilon)=0.128>1/10$.} 
We use a more sophisticated approach, showing that the bias goes down ``on average'': 
We define a \emph{potential function} $q$, which is positive and bounded on $[0,1]$. 
Instead of showing (\ref{wrong}) we show that for any $\delta>0$ there is some $0 \le \theta <1$ such that if $\delta\le\max\{ |\delta_{a}|, |\delta_b|\}$ then
\begin{equation}\label{eq:main}
|\delta_c|q(c) \le \theta \max\{|\delta_a|q(a),|\delta_b|q(b)\}.
\end{equation}
and if $\delta>\max\{ |\delta_{a}|, |\delta_b|\}$ then (\ref{eq:main}) holds for $\theta=1$.
Since $q$ is bounded, this implies that for any arbitrarily small constant $\delta>0$ the bias of any formula becomes $O(\delta)$ after a constant number of computation steps. We can then proceed as above.

We give the main proof in Section \ref{sec:mainproof}. In Section \ref{sec:proof} we prove (\ref{eq:main}), in the main Lemma \ref{lem:bias}. Section \ref{sec:discussion} contains  some remarks on our particular choice of $q$.

\section{Definitions}
\label{Sec:prelim}
A \emph{circuit} is composed of gates. Each \emph{gate} has a certain number of input wires, which is called the \emph{fan-in} of the gate. The wires can take boolean values $0$ or $1$. A gate computes an output bit as a boolean function of its input bits. A \emph{formula} is a particular type of circuit in which the gates are connected in a tree, with the output gate at the root and the input bits at the leaves. In particular, this means that each gate has exactly one output wire.  

A (perfect) PARITY-gate with input bits $x_1$ and $x_2$ outputs $0$ if $x_1=x_2$ and $1$ otherwise. A (perfect) OR-gate outputs $0$ if $x_1=x_2=0$ and $1$ otherwise. 

We call a gate \emph{$\epsilon$-noisy} if it outputs the correct result with probability $1-\epsilon$ and with probability $\epsilon$ it outputs the opposite.
We say that a formula $F$ with noisy gates computes the function $f$ with bias $\Delta>0$ if for all $x\in f^{-1}(0)$, $y\in f^{-1}(1)$: $\mathbb{P}[F(x)=0]\ge \Delta+ \mathbb{P}[F(y)=0]$.\footnote{For our purposes it does not matter that with this definition $f$ and $\bar f$ are actually computed by the same $F$.} 
If $f$ can be computed with some bias $\Delta>0$ we also say that $f$ is \emph{computable with bounded-error}. 

A function $f: \{0,1\}^n \rightarrow \{0,1\}$ \emph{depends}  on the $i$-th input bit $x_i$ if there is some setting of the other bits, such that flipping $x_i$ flips the function value. The number of bits that $f$ depends on is denoted by $d(f)$.

In a formula, we define the \emph{depth} of a wire $A$, denoted by $depth(A)$, as the number of 2-input gates on a path from $A$ to the output wire. Gates with fan-in $1$ are not counted.

For the definition of the quantities $a$ and $\delta_a$ for a wire $A$ we refer to Section \ref{sec:outline}. 
\section{Bias reduction for noisy gates}
\label{sec:proof}
We define the constant $x_0=1/(2-4\e)=(1 + \sqrt{7})/6\approx 0.61$. It will turn out later that an OR-gate with input wires $A,B$ performs best when $a \approx x_0$ and $b \approx x_0$. 
 Further, we define the \emph{potential function} 
\begin{eqnarray}
\nonumber q(x)&=&\left(\tfrac{29}{2}+2 \sqrt{7}\right)\left(x-\tfrac{1}{2}\right)^4\\
\label{q}&&~~+\left(\tfrac
   {5 \sqrt{7}}{2}-{ \tfrac{13}{4}} \right)
   \left(x-\tfrac{1}{2}\right)^2-\tfrac{\sqrt{7}}{2}+\tfrac{
   73}{32}\\
\nonumber &\approx&\text{19.79} (x-\text{0.5})^4+\text{3.36}
   (x-\text{0.5})^2+\text{0.96}.
\end{eqnarray}
\begin{figure}
\label{fig}
\center\includegraphics[width=8.5cm]{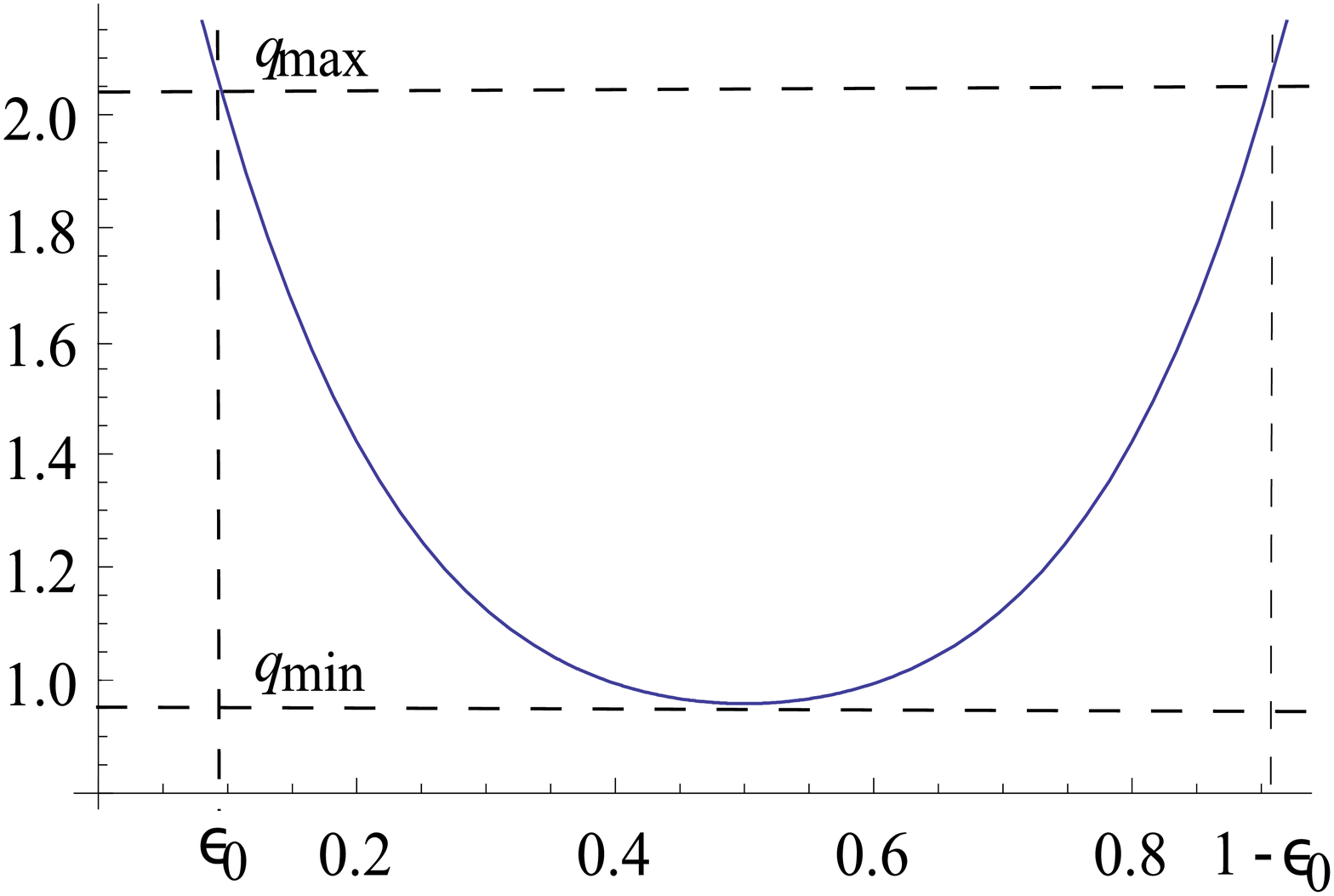}
\caption{Graph of $q(x)$}
\end{figure}
This is a biquadratic function in $(x-1/2)$. Further, $q$ is symmetric around $1/2$ and convex. In $[\e,1-\e]$ $q$ is bounded between $q_{min}=q(1/2)=-\sqrt{7}/{2}+{73}/{32}>0.9$ and $q_{max}=q(\e)=(247 + 8\sqrt{7})/128 <2.1$, see Figure \ref{fig}.

For any $\epsilon \le 1/2$ we define the function
$$\eta_\epsilon(x)=(1-2\epsilon)x + \epsilon.$$
If $x$ is the probability that some variable is $0$, then $\eta_\epsilon(x)$ is the probability that it is $0$ after it has gone through an $\epsilon$-noisy channel.

\subsection{Technical Lemmas}
In the rest of this section we establish inequality (\ref{DPI}) in Lemma \ref{lem:bias}, from which the  proof of the main theorem will follow relatively straightforwardly. The proof of this inequality is quite technical and so at first reading the reader might just want to read the statement of Lemma \ref{lem:bias} and then move immediately to Section \ref{sec:mainproof}, where we establish the main result.
Inequality (\ref{DPI})  can also be checked with the help of a computer (e.g. using Mathematica \cite{m}), but in the remainder of this section we will prove it rigorously.


\begin{myprop}\label{prop:1}
For all $a,b$ with $\e \le a,b \le 1-\e$ it holds that
\begin{eqnarray}
\label{one0}
\nonumber&q(a)q(b) - (1 - 2\e)(aq(a) + bq(b))q(\eta_\e(ab))&\\
&\ge 0.~~~~~~~~~~~~~~~~~~~~~~~~~~~~~~~~~~~~~~~~~&
\end{eqnarray}
\end{myprop}
\begin{proof}
We write $a=x_0+s_a$ and $b=x_0+s_b$. W.l.o.g.\ let $|s_b|\ge|s_a|$ and choose $-1 \le k \le 1$ s.t. $s_a=ks_b$. Then the lhs of (\ref{one0}) can be written as 
\begin{equation}
\label{taylor}
\sum_{i=0}^{11}r_i(k)s_b^{i+2}. 
\end{equation}
The reason why (\ref{taylor}) only starts with a quadratic term in $s_b$ is our special choice of $q$, see Section \ref{choice} for more on this. The first coefficient is easily computed
\begin{eqnarray*}
r_0(k)=\left(3-\tfrac{3\sqrt{7}}{4}\right) \left(k^2+1\right).
\end{eqnarray*}
This function attains its minimum value of $3-3\sqrt{7}/4\approx 1.02$ at $k=0$. Therefore, there is a $\kappa>0$ s.t.\ for $a,b \in [x_0-\kappa,x_0+\kappa]$ the lhs of (\ref{one0}) is non-negative. We show that $\kappa=0.02$ is a solution.

The absolute value of the other coefficients for $-1\le k \le 1$ can be bounded by $|r_1(k)|\le 5$, $|r_2(k)|\le 31$, $|r_3(k)|\le 18$, $|r_4(k)|\le 68$, $|r_5(k)|\le 326$ and for all other $|r_i(k)|\le 5000$. Therefore, if $|s_b|\le 1/50$, (\ref{taylor}) is at least 
\begin{eqnarray*}
\begin{array}{l}
s_b^2\left(1.02-5 (0.02)-31(0.02)^2-18(0.02)^3-\dots\right)\\
~~~\ge 0.90s_b^2\ge 0.
\end{array} 
\end{eqnarray*}

This proves the case $x_0-1/50 \le a,b \le x_0+1/50$. For all other $|a-x_0|\ge 1/50$ or $|b-x_0|\ge 1/50$ the proposition follows from Fact \ref{fact:1} with $\mu=0$.
\end{proof}

We now state some bounds on polynomials. They are similar in spirit to (\ref{one0}), with the crucial difference that these bounds are not tight. This is convenient, because there are several techniques for finding global optima of multivariate polynomials up to arbitrary precision. See \cite{polynomials} for an overview.
 We have used the computer algebra program Mathematica \cite{m}. We used an accuracy of $10^{-10}$ and rounded the results in such a way that the bounds given \emph{are rigorous}.\footnote{Even more simple, one could bound the first derivatives and check all values of the polynomials on a small enough grid.} 
\begin{myfact}\label{fact:1}
 For all $a,b$ with $\e \le a,b \le 1-\e$  with $|a-x_0|\ge 1/50$ and $0 \le \mu \le \xi: = (1-\e-a)(1-\e-b)$ it holds that
\begin{equation}
\label{one'} 
\begin{array}{l}q(a)q(b) - (1 - 2\e)(aq(a) + bq(b))q(\eta_\e(ab+\mu)) \\
> 0.0003.
\end{array}
\end{equation}
\end{myfact}
\begin{proof}
Notice that $\mu$ only appears in the term $q(\eta_\e(ab+\mu))$. For $0 \le \mu \le \xi$ we notice that by convexity of $q$ and linearity of $\eta_\e$ it follows that $q(\eta_\e(ab+\mu))\le \max\{q(\eta_\e(ab)), q(\eta_\e(ab+\xi))\}$. Thus,   (\ref{one'}) is minimized for $\mu=0$ or $\mu=\xi$. For $\mu=0$ the lhs of (\ref{one'}) is lower bounded by $0.0003$ and for $\mu=\xi$ by $0.01$. 
\end{proof}

\begin{myfact}\label{fact:2a}
 For all $a,b,\mu$ with $\e \le b \le 1-\e$, $1/2 \le a \le 1-\e$ and $|\mu| \le \xi := 2(1-\e-a)(1-\e-b) $ it holds that
\begin{eqnarray*}
\label{fc1}
\begin{array}{cl}
q(a)q(b)-
((2a-1)q(a)+(2b-1)q(b))(1-2\e) \\~~~~~~~~~~~~~\times q(\eta_\e(ab +(1-a)(1-b) + \mu))\\
\ge  0.45.~~~~~~~~~~~~~~~~~~~~~~~~~~~~~~~~~~~~~~~~~~~ 
\end{array}
\end{eqnarray*}
\end{myfact}
\begin{proof}
For $\mu = \xi$ the term is lower bounded by $0.48$ and for $\mu = -\xi$ by $0.55$. Using convexity of $q$ as above the fact follows.
\end{proof}

\begin{myfact}\label{fact:2b}
 For all $a,b,\mu$ with $\e \le b \le 1-\e$, $\e \le a \le 1/2$ and $|\mu| \le \xi := 2(a-\e)(1-\e-b) $ it holds that
\begin{eqnarray*}
\label{fc2}
\begin{array}{cl}
q(a)q(b)-
((1-2a)q(a)+(2b-1)q(b))(1-2\e)\\~~~~~~~~~~~~~~~\times q(\eta_\e(ab +(1-a)(1-b) + \mu))\\
\ge 0.48.~~~~~~~~~~~~~~~~~~~~~~~~~~~~~~~~~~~~~~~~~~~~ 
\end{array}
\end{eqnarray*}
\end{myfact}
\begin{proof}
For $\mu = \xi$ the term is lower bounded by $0.51$ and for $\mu = -\xi$ by $0.48$. The fact then follows by convexity of $q$ as above.
\end{proof}

\begin{myfact}\label{fact:3}
Let $a,b,\mu$ with $\e \le a,b \le 1-\e$ Then 
\begin{eqnarray*}
q(a)-(1-2\e)bq(\eta_\e(ab-\mu)) &>& 0.22~~
\end{eqnarray*}holds
if (a) $a\le 1/2$ and  $-(a-\e)(1-\e-b) \le \mu \le 0$  or (b) $ 1/2 \le a $ and  $-(1-\e-a)(1-\e-b) \le \mu \le 0 $.
\end{myfact}
\begin{proof}
For $\mu=0$ (and all $\e \le a \le 1-\e$) the term is lower bounded by $0.23$. For both cases $a \le 1/2$ ,  $\mu=-(a-\e)(1-\e-b)$  and $ 1/2 \le a $, $\mu= -(1-\e-a)(1-\e-b)  $ the term is lower bounded by  $0.22$. Using convexity of $q$ as above the fact follows.
\end{proof}

We can state our main Lemma.
\begin{mylemma}
\label{lem:bias}
Let $\e \le \epsilon \le 1/2$.
 Assume an $\epsilon$-noisy OR-gate or PARITY-gate, with input wires $A$ and $B$ and output wire $C$. Let
\begin{eqnarray}
\label{probs}
\begin{array}{rcccl}
 \e &\le &\mathbb{P}[A=0 \mid x_i=0]& \le& 1-\e\\
\e& \le &\mathbb{P}[A=0 \mid x_i=1]& \le& 1-\e,
\end{array}
\end{eqnarray}
and let the same be true for $B$.
Define $a,b,c$ and $\delta_a,\delta_b,\delta_c$ for $A,B,C$ as in Section \ref{sec:intro}.
\begin{enumerate}
 \item\label{equal}
The following inequality holds for $\theta=1$:
\begin{equation}
\label{DPI}
|\delta_c|q(c) \le \theta \max\{|\delta_a|q(a),|\delta_b|q(b)\}.
\end{equation}
\item\label{strict} For any $\delta>0$ there is a $0\le \theta <1 $ such that if  $|\delta_a|\ge\delta$ or $|\delta_b|\ge\delta$, (\ref{DPI}) is still true for this $\theta$.
\end{enumerate}
\end{mylemma}
\begin{proof}We consider the OR-gate first.
We have $\mathbb{P}[C=0 \mid x_i=0]=\eta_\epsilon\left(\left(a+\delta_a/2\right)\left(b+\delta_b/2\right)\right)$ and $\mathbb{P}[C=0 \mid x_i=1]=\eta_\epsilon\left(\left(a-{\delta_a}/{2}\right)\left(b-{\delta_b}/{2}\right)\right)$, which implies
\begin{eqnarray*}
\delta_c&=&\left(a\delta_b+b\delta_a\right)(1-2\epsilon)\\
c&=&\eta_\epsilon\left(ab+{\delta_a\delta_b}/{4}\right).
\end{eqnarray*}
Increasing $\epsilon$ decreases $|\delta_c|$ as well as also $q(c)$, since $c$ gets closer to $1/2$ and $q$ decreases towards $1/2$. Thus we may assume $\epsilon=\e$. Further, we may assume $|\delta_a|q(a)\ge |\delta_b|q(b)$. Note that, for $\delta_a=0$ we then also have $\delta_b=0$ and the Lemma holds trivially. In the remainder we therefore assume $\delta_a \neq 0$. In fact, we will even assume $\delta_a > 0$: In case $\delta_a < 0$ we can just formally replace every occurrence of $\delta_a$ and $\delta_b$ with $-\delta_a$ resp. $-\delta_b$. Because of the absolute value signs, this will not change the validity of (\ref{Yeah}). 
So we have to prove
\begin{eqnarray}\label{Yeah}
&&(1-2\epsilon)\left|a\delta_b+b\delta_a\right|q\left(\eta_\epsilon\left(ab+\delta_a\delta_b/4\right)\right) \\
\nonumber &&\le \theta |\delta_a|q(a).
\end{eqnarray}
 In the remainder, we will repeatedly use that $a$ and $b$ are bounded between $\e$ and $1-\e$ and that in this range, $0.9 < q_{min} \le q(a) \le q_{max} < 2.1$, without mentioning it each time. We distinguish the following cases:

$\mathbf{\delta_b> 0}$: Since we assumed $|\delta_a|q(a) \ge |\delta_b|q(b)$, it is enough to prove (\ref{Yeah}) where we replace the first occurrence of $\delta_b$ by $\delta_a{q(a)}/{q(b)}$. Cancelling $\delta_a$ and multiplying by $q(b)$ we get
\begin{equation}
\label{Bern}
\begin{array}{l}
\theta q(a)q(b)\\
~~- (1-2\e)\left(aq(a)+bq(b)\right)q\left(\eta_\e\left(ab+\delta_a\delta_b/4\right)\right)\\
\ge 0.
\end{array}
\end{equation} 

In case $|a-x_0|\ge 1/50$ or $|b-x_0|\ge 1/50$, note that $\delta_a\delta_b/4 \le (1-\e-a)(1-\e-b)$. If we set $\mu=\delta_a\delta_b/4$ and $\theta=1$, then by Fact \ref{fact:1} the lhs of (\ref{Bern}) is greater than $0.0003$. This implies the existence of a $\theta<1$ for (\ref{Bern}) and settles both parts of the Lemma.

We are left with the case $|a-x_0|< 1/50$ and $|b-x_0|< 1/50$.  By (\ref{probs}) we can then bound $\delta_a/2\le 1-\e -a\le 1-\e-x_0+1/50<0.33$ and similarly  $\delta_b/2<0.33$, i.e.\ $(1-2\e)\delta_a\delta_b/4 <0.1$. 
We also note that in our case $0.37< \eta_\e(ab)<0.42$. 
By convexity, $\min_{0.37\le x \le 0.42}q(x)-q(x+0.1)=q(0.42)-q(0.52)>0.02$, and thus
$q(\eta_\e(ab)+0.1)<q(\eta_\e(ab)) - 0.02$. This last inequality, convexity of $q$ and $(1-2\e)\delta_a\delta_b/4 <0.1$ imply $q(\eta_\e(ab)+(1-2\e)\delta_a\delta_b/4)<q(\eta_\e(ab)) - \frac{0.02}{0.1}(1-2\e)\delta_a\delta_b/4$. Noting that $\eta_\e(ab)+(1-2\e)\delta_a\delta_b/4=\eta_\e(ab + \delta_a\delta_b/4)$
this becomes
\begin{eqnarray}
 \label{better}
&&q(\eta_\e(ab+\delta_a\delta_b/4))\\
\nonumber &&<q(\eta_\e(ab))-(1-2\e)\delta_a\delta_b/20.
\end{eqnarray}
In particular $q(\eta_\e(ab+\delta_a\delta_b/4))<q(\eta_\e(ab))$. Plugging the lhs of this into (\ref{Bern}) and using Proposition \ref{prop:1} implies  (\ref{Bern}) for $\theta=1$. This establishes part \ref{equal} of the Lemma for $\delta_b>0$.

Now part \ref{strict} of the Lemma. Let $\delta_a\ge \delta$ or $\delta_b\ge\delta$. 
Consider first the case that $\delta_b$  is not too small compared to $\delta_a$, say $\delta_b \ge \delta_a/100$. Together with our assumption $|\delta_a|q(a)\ge |\delta_b|q(b)$ this implies $(1-2\e)\delta_a\delta_b/20 \ge (1-2\e)\delta^2/2000$. With (\ref{better}) we then get 
$q(\eta_\e(ab+\delta_a\delta_b/4))+ c <q(\eta_\e(ab))$ for  $c=(1-2\e)\delta^2/2000>0$ and putting this into (\ref{one0}) gives 
$q(a)q(b) - (1 - 2\e)(aq(a) + bq(b))(q(\eta_\e(ab+\delta_a\delta_b/4))+ c) > 0.$ This
implies the existence of a $\theta<1$ for (\ref{Bern}) and establishes part 2 of the Lemma when  $\delta_b \ge \delta_a/100$.

If $\delta_b$ is small, i.e.\ $\delta_b < \delta_a/100$, then upper bounding the first occurrence of $\delta_b$ by $\delta_aq(a)/q(b)$ to get from (\ref{Yeah}) to (\ref{Bern}) was far from tight. A better bound is $\delta_b < \delta_aq(a)/(10q(b))$, which derives from $q(a)/(10q(b))\ge q_{min}/(10q_{max})>1/100$. Analogously to the derivation of (\ref{Bern}) we get
\begin{equation}
\label{special}
\begin{array}{l}
\theta q(a)q(b) \\
-(1-2\e)\left({aq(a)}/{10}+bq(b)\right)q\left(\eta_\e\left(ab+{\delta_a\delta_b}/{4}\right)\right)\\
\ge 0. 
\end{array}
\end{equation} 
By (\ref{probs}), $a> \e$. 
Thus, $aq(a)>\e q_{min}$ and also $q\left(\eta_\e\left(ab+\delta_a\delta_b/4\right)\right) >q_{min}$. Hence, the lhs of (\ref{special}) is at least $(1-2\e)\e q_{min}^29/10$ smaller than the lhs of (\ref{Bern}). Since we already proved earlier that (\ref{Bern}) holds for $\theta=1$ without the restriction $\delta_b < \delta_a/100$, we conclude that (\ref{special}) holds for some $\theta<1$.  This establishes part \ref{strict} of the Lemma for $\delta_b < \delta_a/100$.

$\mathbf{\delta_b \le 0}$: It is enough to prove (\ref{Yeah}) where we replace $|a\delta_b+b\delta_a|$ by (a) $|b\delta_a|$ or (b) $|a\delta_b|$. If in case (a) we cancel $\delta_a$ and $q(a)$ after the replacement, we see that a $\theta<1$ must exist if 
\begin{equation}
\label{p1}
q(a)-(1-2\e)bq(\eta_\e(ab+\delta_a\delta_b/4)) \ge \chi ,
\end{equation}
for some $\chi>0$. Note that in case $a\le 1/2$ we have  $-(a-\e)(1-\e-b) \le \delta_a\delta_b/4 \le 0$ and in case $ 1/2 \le a $ we have   $-(1-\e-a)(1-\e-b) \le \delta_a\delta_b/4 \le 0 $. The Lemma then follows from Fact \ref{fact:3}.

For case (b) we note that $|a\delta_b|\le a\delta_aq(a)/q(b)$. Replacing $|a\delta_b+b\delta_a|$ in (\ref{Yeah}) by  $a\delta_aq(a)/q(b)$ and rearranging terms we get exactly the same as (\ref{p1}), with $a$ and $b$ swapped. We proceed as in case (a).

We now consider the PARITY-gate. First note, that if the two input wires of a noiseless PARITY gate are independently $0$ with probability $\alpha$ resp. $\beta$, then the output wire will be $0$ with probability $\alpha\beta+(1-\alpha)(1-\beta)$. Thus, in our case
$$\begin{array}{lcl}
&&\mathbb{P}[C=0 \mid x_i=0]\\
&=&
\eta_\epsilon(\left(a+{\delta_a}/{2}\right)\left(b+{\delta_b}/{2}\right)\\&&~+\left(1-a-{\delta_a}/{2}\right)\left(1-b-{\delta_b}/{2}\right))\\
\mbox{and}\\ 
&&\mathbb{P}[C=0 \mid x_i=1]\\
&=&\eta_\epsilon(\left(a-{\delta_a}/{2}\right)\left(b-{\delta_b}/{2}\right)\\&&~+\left(1-a+{\delta_a}/{2}\right)\left(1-b+{\delta_b}/{2}\right))
\end{array}$$
 which implies
\begin{eqnarray*}
c&=&\eta_\epsilon\left(ab+(1-a)(1-b)+{\delta_a\delta_b}/{2}\right)\\
\delta_c&=&\left((2a-1)\delta_b+(2b-1)\delta_a\right)(1-2\epsilon)
\end{eqnarray*}
We need to prove
\begin{eqnarray}
\label{pari}
\begin{array}{cl}
&\left|(2a-1)\delta_b+(2b-1)\delta_a\right|(1-2\epsilon)\times \\
&~~~~~~q(\eta_\epsilon\left(ab+(1-a)(1-b)+{\delta_a\delta_b}/{2}\right)) \\
&\le \theta |\delta_a| q(a).
\end{array}
\end{eqnarray}
As for the OR-gate we only need to consider $\epsilon=\e$ and may assume $\delta_a\ge 0$ w.l.o.g, because otherwise we can just change the signs of both $\delta_a$ and $\delta_b$. Also, w.l.o.g. we assume $|\delta_a|q(a)\ge |\delta_b|q(b)$. If $\delta_a=0$, then also $\delta_b=0$ and the Lemma becomes trivial. So we assume $\delta_a>0$. 
Further, we may assume $b\ge 1/2$ (and therefore  $(2b-1)\delta_a\ge 0$), because formally replacing  $a$ and $b$ by $1-a$ and $1-b$ does not change (\ref{pari}).
We condition on the sign of $2a-1$.

First ${2a-1 \ge 0}$. It is enough to prove (\ref{pari}), where we replace the first occurrence of $\delta_b$ by $\delta_aq(a)/q(b)$, since we assumed $|\delta_a|q(a)\ge |\delta_b|q(b)$. Cancelling $\delta_a$ and rearranging terms, the existence of a $0\le \theta<1$ for (\ref{pari}) then follows from 
\begin{eqnarray*}
\label{c1}
\begin{array}{cl}
q(a)q(b)-
((2a-1)q(a)+(2b-1)q(b))(1-2\e) \\~~~~~~~~~\times q(\eta_\e(ab +(1-a)(1-b) + \delta_a\delta_b/2))\\
\ge  \chi>0.~~~~~~~~~~~~~~~~~~~~~~~~~~~~~~~~~~~~~~~~~~~~~~~ 
\end{array}
\end{eqnarray*}
This inequality follows from Fact \ref{fact:2a} by noting that $|\delta_b|\le 2(1-\e-b)$ and $|\delta_a|\le 2(1-\e-a)$. 

In case  ${2a-1 < 0}$ we can proceed similarly, where this time we replace the first occurrence of $\delta_b$ by $-\delta_aq(a)/q(b)$ and bound $|\delta_a|\le 2(a-\e)$. The resulting inequality follows from Fact \ref{fact:2b}.
\end{proof}

\section{Proof of Theorem \ref{the:main}}
\label{sec:mainproof}
\begin{proof}
Let $f$ be any function and let $F$ be any formula with noisy gates that fail independently with probability at least $\e$. Let $F$ compute $f$ with bias $\Delta$. We show that $f$ depends on at most a constant number of bits, i.e.\ $d(f) \le c(\Delta)$, for some function $c(\Delta)$. 

Before starting we note the following: Every $\epsilon$-noisy fan-in-2 gate can be constructed from an $\epsilon$-noisy PARITY- or an $\epsilon$-noisy OR-gate, perfect NOT-gates and constant inputs. Hence, we may assume  that $F$ is constructed from perfect NOT-gates and noisy PARITY-gates and OR-gates. 

Let $x_i$ be an input bit on which $f$ depends with the additional property that any input wire of $F$ carrying $x_i$ has depth at least $\left\lceil \log_2 d(f) \right\rceil$. Because all gates in $F$ have fan-in at most $2$, the existence of such $x_i$ is guaranteed. Fix all other input bits such that the output of $F$ changes when flipping $x_i$.

Set $\mathrm{D}=\left\lceil \log_2 d(f)\right\rceil-1$ and $\delta=\frac{\Delta}{2q_{max}}$. Let $\theta<1$ be given by Lemma \ref{lem:bias} for this $\delta$. In case this results in $\theta<1-2\e$, set $\theta=1-2\e$. (The adjustment $\theta\ge 1-2\e$ is not really needed, but will later simplify the proof.) We will prove inductively that for any wire $C$ at depth $d\le \mathrm{D}$
\begin{equation}
 \label{invariant}
q(c)|\delta_c|\le \max\{\tfrac{\Delta}{2},\theta^{\mathrm{D}-d}q_{max}\}.
\end{equation}

For $d=\mathrm{D}$ (\ref{invariant}) holds trivially.  Now take any wire $C$ in $F$ with depth $d < \mathrm{D}$. We distinguish what signal $C$ carries.

Firstly, $C$ can be an input wire carrying $x_j$. Then necessarily $i \neq j$, because input wires carrying $x_i$ have depth at least $\mathrm{D}+1$. Thus, $\delta_c=0$ and (\ref{invariant}) holds. 

Secondly, $C$ can be the output of a noiseless NOT-gate, which has input wire $B$. Note that since we do not count NOT-gates in the depth of a wire, $depth(C)=depth(B)$, $c=1-b$ and $\delta_c=-\delta_b$. Then, by symmetry of $q$ around $1/2$ we get (\ref{invariant}) for $C$ from the same statement for $B$. 

Thirdly, $C$ can be the output of gate $G$, with $G$ either an OR-gate or a PARITY-gate. Let the input wires to $G$ be $A$ and $B$.
If one wire is a constant, say $A$, then gate $G$ is essentially a (noisy) gate with fan-in 1. Hence, $G$ always outputs either a (noisy) $0$ or $1$, or $G$ is the noisy identity- or the noisy NOT-gate. In the first two cases $\delta_c=0$. In the last two cases we can easily calculate that $|b-1/2|(1-2\epsilon)=|c-1/2|$ and $|\delta_c|\le (1-2\e) |\delta_b|$. Because $q$ decreases monotonically towards $1/2$ and we chose  $\theta\ge 1-2\e$, (\ref{invariant}) holds.

So we are left with the case where both inputs to $G$ are non-constant. Since $d< \mathrm{D}$, both wires $A$ and $B$ are the output of some noisy gate, so the conditions (\ref{probs}) in Lemma \ref{lem:bias} are satisfied. We may assume $|\delta_b|q(b) \le |\delta_a|q(a)$ w.l.o.g. If $|\delta_a|q(a)\le \Delta/2$, then by part \ref{equal} from Lemma \ref{lem:bias} also $|\delta_c|q(c)\le \Delta/2$ and (\ref{invariant}) holds. If $|\delta_a|q(a)> \Delta/2$, then $|\delta_a|>\frac{\Delta}{2q_{max}}=\delta$. Then (\ref{invariant}) follows from part \ref{strict} of Lemma \ref{lem:bias} and the inductive assumption.

Let $O$ be the output wire of $F$, which by assumption has bias $\Delta$. Because $q(o)\Delta \le \Delta /2$ is impossible (since $q(o)\ge q_{min} > 1/2$) we get from (\ref{invariant}): $q(o) \Delta \le \theta^{\mathrm{D}}q_{max}$, and further $ \Delta \le \theta^{\left\lceil \log_2 d(f)\right\rceil-1}({q_{max}}/{q_{min}})$, which implies
$\frac{\log_2(\Delta q_{min}/q_{max})}{\log_2 \theta}+1\ge \log_2d(f)$. Since $\theta$ depends only on $\Delta$, $d(f)$ is upper bounded by the function 
$$c(\Delta):=2\left({\Delta q_{min}}/{q_{max}}\right)^{1/\log_2 \theta}.$$
\end{proof}

\section{Discussion}
\label{sec:discussion}
We have shown a tight threshold for the noise which is tolerable for computation by formulas with gates of fan-in at most 2. This is the first result for gates with an even number of wires. It should be possible to generalize it to other fan-in, although the proof is probably more tedious. 

The same bound probably also applies to \emph{circuits} with gates of fan-in at most $2$. 

\subsection{Choice of potential function}\label{choice}
So far we have not given any idea of why we chose this particular potential function. In fact, this choice is not unique. The choice of $q$ was determined as follows: (1) It is convenient to choose $q$ symmetric around $1/2$, so applying a NOT-gate to wire $A$ does not change the value of $|\delta_a|q(a)$. (2) It is natural to scale $q$ such that $q(x_0)=1$. 
(3) After these choices,  we have to choose $\frac{d}{dx}q(x)|_{x=x_0}=\frac{1}{2}(-1+\sqrt{7})$. This ensures that (\ref{taylor}) does not have a linear term in $s_b$ and only starts with the quadratic term, i.e.\ ``$r_{-1}(k) \equiv 0$''. (4) We also need $\frac{d^2}{dx^2}q(x)|_{x=x_0}> 16-4\sqrt{7} \approx 5.42$, because that makes $r_0(k)>0$ for $-1\le k \le 1$. 
The rest of the choices are not so binding. 

However, a quadratic function alone is not enough.  For (\ref{taylor}) to be at least $0$ one also has constraints on higher derivatives of $q$. The expression in (\ref{q}) for $q$ is one of the ``nicer'' possible potential functions. One can also find a possible  $q$, by dividing the interval $[\e,1-\e] $ into smaller intervals and define $q$ as different quadratic functions in each of these intervals.



\section{Acknowledgements}
I would like to thank Intel chip designers Bob Colwell and Shekhar Borkar for pointing out \cite{s:noisy} and \cite{Bobgodwell}. I am also grateful to Ronald de Wolf, Peter Harremo\"es and an anonymous referee for some comments and to Ben Reichardt for proof-reading and pointing out some errors.

\bibliographystyle{abbrv}

\end{document}